\documentclass[a4paper]{article}
\usepackage[margin=1in]{geometry}
\usepackage{microtype}%if unwanted, comment out or use option "draft"
\usepackage{graphicx}
\usepackage[algoruled,linesnumbered,algo2e,vlined]{algorithm2e}
\usepackage[final,mode=multiuser]{fixme}
\usepackage{amsmath,amssymb}
\usepackage{url}
\usepackage{rotating}
\usepackage{comment}
\usepackage{multirow}
\usepackage{amsthm}
\theoremstyle{plain}
\newtheorem{theorem}{Theorem}
\newtheorem{definition}[theorem]{Definition}
\newtheorem{lemma}[theorem]{Lemma}
\newtheorem{corollary}[theorem]{Corollary}
\newtheorem{observation}[theorem]{Observation}

\newcommand{\floor}[1]{\left \lfloor #1 \right \rfloor}
\newcommand{\ceil}[1]{\left \lceil #1 \right \rceil}
\newcommand{\lcp}{\mathsf{lcp}}
\newcommand{\LCE}{\mathsf{LCE}}
\newcommand{\LCP}{\mathsf{LCP}}
\newcommand{\SA}{\mathsf{SA}}

\newcommand{\rmq}{\mathsf{rmq}}
\newcommand{\ISA}{\mathsf{ISA}}
\newcommand{\SSA}{\mathsf{SSA}}
\newcommand{\SLCP}{\mathsf{SLCP}}
\newcommand{\Meta}[3]{\hat{#1}_{#2,#3}}
\newcommand{\CA}[1]{\mathsf{CA}_{#1}}
  
\fxuseenvlayout{colorsig}
\fxusetargetlayout{color}

\FXRegisterAuthor{yt}{ayt}{YT}
\FXRegisterAuthor{si}{asi}{SI}
\FXRegisterAuthor{hb}{ahb}{HB}
\FXRegisterAuthor{mt}{amt}{MT}
\fxsetface{env}{}

\title{
  Deterministic sub-linear space LCE data structures with efficient construction
}
\author{
  Yuka~Tanimura$^1$\quad
  Tomohiro~I$^2$\quad
  Hideo~Bannai$^1$\quad
  Shunsuke~Inenaga$^1$\quad\\
  Simon~J.~Puglisi$^3$\quad
  Masayuki~Takeda$^1$\\
  {$^1$ Department of Informatics, Kyushu University, Japan}\\
  {\texttt{\{yuka.tanimura,bannai,inenaga,takeda\}@inf.kyushu-u.ac.jp}}\\
  {$^2$ Kyushu Institute of Technology, Japan}\\
  {\texttt{tomohiro@ai.kyutech.ac.jp}}\\
  {$^3$ Department of Computer Science, University of Helsinki, Finland}\\
  {\texttt{puglisi@cs.helsinki.fi}}
}

\date{}
\begin{document}
\maketitle
\begin{abstract}
  Given a string $S$ of $n$ symbols, a {longest common extension query} $\LCE(i,j)$ asks for 
  the length of the longest common prefix of the $i$th and $j$th suffixes of $S$. LCE queries 
  have several important applications in string processing, perhaps most notably to suffix sorting. 
  Recently, Bille et al. (J. Discrete Algorithms 25:42-50, 2014, Proc. CPM 2015: 65-76)
  described several data structures for answering LCE queries that offers a space-time trade-off between
  data structure size and query time.
  In particular, for a parameter $1 \leq \tau \leq n$,
  their best deterministic solution is a data structure of size $O(n/\tau)$
  which allows LCE queries to be answered in $O(\tau)$ time.
  However, the construction time for all deterministic versions of their data structure
  is quadratic in $n$.
  In this paper, we propose a deterministic solution that achieves a similar space-time trade-off
  of $O(\tau\min\{\log\tau,\log\frac{n}{\tau}\})$ query time using $O(n/\tau)$ space,
  but significantly improve the construction time to $O(n\tau)$. 
\end{abstract}

\section{Introduction}

Given a string $S$ of $n$ symbols, a {\em longest common extension query} $\LCE(i,j)$ asks for
the length of the longest common prefix of the $i$th and $j$th suffixes of $S$. 

The ability to efficiently answer LCE queries allows optimal solutions to many string processing problems. 
Gusfield's book~\cite{Gusfield97}, 
for example, lists several applications of LCEs to basic pattern matching and discovery problems, including: 
pattern matching with wildcards, mismatches and errors;
the detection of various types of palindromes (maximal, complimented, separated, approximate); and
the detection of repetitions and approximate repetitions.
Lempel-Ziv parsing~\cite{KKP13} and suffix sorting~\cite{Karkkainen:2006je,Karkkainen07} are 
two more fundamental string processing problems to which LCEs are key.

Without preprocessing, answering an arbitrary query $\LCE(i,j)$ requires $O(n)$ time: we simply compare the
suffixes starting at positions $i$ and $j$ character by character until we find a mismatch.
To answer queries faster we could build the suffix tree and preprocess it for lowest-common-ancestor 
queries. This well-known solution answers queries in $O(1)$ time and the data structure is of $O(n)$ 
size and takes $O(n)$ time to construct. 

In recent years, motivated by scenarios where $O(n)$ space is prohibitive, several authors have sought 
data structures that achieve a trade-off between data structure size and query time. The best trade-off 
to date is due to Bille et al.'s~\cite{bille15:_longes_common_exten_sublin_space}, who describe 
a data structure of size $O(n/\tau)$ which allows LCE queries to be answered in $O(\tau)$ time.
However, as described in~\cite{bille15:_longes_common_exten_sublin_space}, 
their data structure requires $O(n^2)$ time to construct if only $O(n/\tau)$ working space is allowed. 
This is a major drawback, because it does not allow the space/query time trade-off to be passed on to 
applications --- indeed, construction of the data structure would become a time bottleneck in all the 
applications listed above.

\medskip

The main contributions of this article are as follows:
\begin{enumerate}
\item We describe a new data structure for LCEs that has size $O(n/\tau)$, query time
$O(\tau\log\tau)$, and, critically, can be constructed in $O(n\tau)$ time.
\item We show how to combine the new data structure with one of Bille et al.'s to derive 
a structure that has $O(\tau\log\frac{n}{\tau})$ query time and the same space and construction
bounds as the new structure. As a side result, we also show how this particular structure of Bille 
et al. can be constructed efficiently.
\end{enumerate} 
Table~\ref{table:resultcomparison} summarizes our results and previous work on the problem.

In the next section we lay down notation and some basic algorithmic and data structural tools. 
Then, in Section~\ref{sec:new}, we introduce our new LCE data structures, beginning with a 
a slightly modified version of one of Bille et al.s data structures, followed by the new 
and combined data structures. Section~\ref{sec:construct} deals with efficient construction. 
We finish, in Section~\ref{sec:appl}, by noting that our new structures lead directly to 
improved (deterministic) bounds for the sparse suffix sorting problem.

\begin{table}
  \begin{center}
    \caption{Deterministic solutions to LCE.}
    \label{table:resultcomparison}      
    \begin{tabular}{|c|c|c|c|c|p{3cm}|}\hline
      \multicolumn{2}{|c|}{Data Structure} & \multicolumn{2}{|c|}{Preprocessing} & \multirow{2}{*}{Trade-off range} & \multirow{2}{*}{Reference} \\\cline{1-4}
      Space & Query & Space & Time & &\\\hline
      $1$ & $n$ & 1 & 1 & - & na\"ive computation\\ \hline
      $n$ & 1 & $n$ & $n$ & - & suffix array + RMQ\\ \hline
      $\frac{n}{\tau}$ & $\tau^2$ & $\frac{n}{\tau}$ & $\frac{n^2}{\tau}$ & $1 \leq \tau \leq \sqrt{n}$ & \cite{bille14:_time} \\\hline
      %% Bille et al. orig.
      $\frac{n}{\tau}$ & $\tau\log^2\frac{n}{\tau}$ & $\frac{n}{\tau}$ & $n^2$ & $1 \leq \tau \leq n$ & \cite{bille15:_longes_common_exten_sublin_space}, Section 2\\\hline
      $\frac{n}{\tau}$ & $\tau$ & $\frac{n}{\tau}$ & $n^{2+\epsilon}$ & $1 \leq \tau \leq n$ & \cite{bille15:_longes_common_exten_sublin_space}, Section 4\\
      \hline
      %% Bille et al. improved
      $\frac{n}{\tau}$ & $\tau\log^2\frac{n}{\tau}$ & $\frac{n}{\tau}$ & $n\tau + n\log \frac{n}{\tau}$ & $1 \leq \tau \leq n$ & Improved preprocessing for~\cite{bille15:_longes_common_exten_sublin_space}. This work.\\
      \hline
      %% new method
      $\frac{n}{\tau}$ & $\tau\log\tau$ & $\frac{n}{\tau}$ & $n\tau$ &  $1 \leq \tau \leq \frac{n}{\log n}$  & This work.\\\hline
      %% combination
      $\frac{n}{\tau}$ & $\tau\log\frac{n}{\tau}$ & $\frac{n}{\tau}$ & $n\tau$ & $1 \leq \tau \leq n$ & This work.\\\hline
      $\frac{n}{\tau}$ & $\tau\min\{\log\tau,\log\frac{n}{\tau}\}$ & $\frac{n}{\tau}$ & $n\tau$ & $1 \leq \tau \leq n$ & This work.\\\hline
    \end{tabular}
  \end{center}
\end{table}

\section{Preliminaries}
Let $\Sigma = \{ 1, \ldots, \sigma\}$ denote the alphabet, and $\Sigma^*$ the set of strings.
If $w = xyz$ for any strings $w,x,y,z$, then $x$,$y$, and $z$ are respectively
called a {\em prefix}, {\em substring}, and {\em suffix} of $w$.
For any string $w$, let $|w|$ denote the length of $w$,
and for any $0 \leq i < |w|$, let $w[i]$ denote the $i$th character,
i.e., $w = w[0]\cdots w[|w|-1]$.
For convenience, let $w[i] = 0$ when $i \geq |w|$.
For any $0 \leq i \leq j$, let $w[i..j] = w[i]\cdots w[j]$,
and for any $0 \leq i < |w|$, let $w[i..]= w[i..|w|-1]$.
We denote $x\prec y$ if string $x$ is lexicographically smaller than string $y$.

For any string $w$, let $\lcp_w(i,j)$ denote the length of the longest common
prefix of $w[i..]$ and $w[j..]$.
We will write $\lcp(i,j)$ when $w$ is clear from the context.
Since $\lcp_w(i,i) = |w|-i$, we will only consider the case when $i\neq j$.
Note that an LCE query $\LCE(i,j)$ is equivalent to computing $\lcp_w(i,j)$.

For any integers $i\leq j$, let $[i..j]$ denote the set of integers
from $i$ to $j$,
and for $0\leq p < \tau$,
let $[i..j]^\tau_p = \{ k \mid k \in [i..j], k \bmod \tau = p\}$.

For any string $w$ of length $n$ and  $0\leq p < \tau$,
let $\Meta{w}{\tau}{p}$ denote a string of length $\ceil{(|w|-p)/\tau}$
over the alphabet $\{ 1,\ldots, \sigma^\tau \}$
such that $\Meta{w}{\tau}{p}[i] = w[p + \tau i .. p + \tau (i+1) - 1]$
for any $i \geq 0$.
We call $\Meta{w}{\tau}{p}$ the {\em meta-string} of $w$ wrt. sampling rate $\tau$ and offset $p$,
and each character of $\Meta{w}{\tau}{p}$ is called a {\em meta-character}.

In the rest of the paper, we assume a polynomialy bounded integer alphabet, i.e.,
for some constant $c\geq 0$, $\sigma = O(n^c)$ for any input string $w$ of length $n$.

\begin{definition}[\cite{manber93:_suffix_array}]
  The suffix array $\SA_w$ of a string $w$ of length $n$
  is an array of size $n$ containing a permutation of $[0..n-1]$
  that represents the lexicographic order of the suffixes of $w$,
  i.e., $w[\SA_w[0]..] \prec \cdots \prec w[\SA_w[n-1]..]$.
  The inverse suffix array $\ISA_w$ is an array of size $n$
  such that $\ISA_w[\SA_w[i]] = i$ for all $0 \leq i < n$.
  The LCP array $\LCP_w$ of a string $w$ of length $n$ is an
  array of size $n$ such that $\LCP_w[0] = 0$ and
  $\LCP_w[i] = \lcp_w(\SA_w[i-1], \SA_w[i])$ for $0 < i < n$.
\end{definition}

\begin{lemma}[\cite{kasai01:_linear_time_longes_common_prefix,Kim:2003br,Ko:2003dt,Karkkainen:2006je}]
  For any string $w$ of length $n$, $\SA_w,\ISA_w,\LCP_w$ can be computed in $O(n)$ time and space.
\end{lemma}

For any array $A$ and $0 \leq i\leq j < |A|$,
let $\rmq_A(i,j)$ denote a Range Minimum Query (RMQ),
i.e., $\rmq_A(i,j) = \arg\min_{k\in[i..j]} \{ A[k] \}$.
It is well known that $A$ can be preprocessed in linear time and space
so that $\rmq_A(i,j)$, for any $0 \leq i \leq j < |A|$,
can be answered in constant time~\cite{bender00:_lca_probl_revis}.
Since $\lcp_w(i,j) = \LCP_w[\rmq_{\LCP_w}(i^\prime+1, j^\prime)]$ where
$i^\prime = \min\{\ISA_w(i),\ISA_w(j)\}$
and $j^\prime = \max \{\ISA_w(i),\ISA_w(j)\}$,
it follows that
a string of length $n$ can be preprocessed in $O(n)$ time and space
so that for any $0 \leq i,j < n$, $\lcp_w(i,j)$ can be computed in $O(1)$ time.

Our algorithm relies on sparse suffix arrays.
For a string $w$ of length $n$ and any set $P \subseteq [0..n-1]$ of positions,
let $\SSA_P[0..|P|-1]$ be an array consisting of entries of $\SA$
that are in $P$, i.e.,
for any $0 \leq i < |P|$, $\SSA_P[i] \in P$, and
$w[\SSA_P[0]..] \prec \cdots \prec w[\SSA_P[|P|-1]..]$.
The sparse LCP array $\SLCP_P[0..|P|-1]$ is defined analogously,
$\SLCP_P[i] = \lcp_w(\SSA_P[i-1],\SSA_P[i])$.

Let $1 \leq \tau \leq n$ be a parameter called the {\em sampling rate}.
When, $P = [0..n-1]^\tau_p$, for some $0 \leq p < \tau\leq n$,
$\SSA_P$ is called an evenly spaced sparse
suffix array with sampling rate $\tau$ and offset $p$.
Given an evenly spaced sparse suffix array $\SSA_P$,
we can compute in $O(\frac{n}{\tau})$ time,
a representation of the sparse inverse suffix array $\ISA_P$
as an array $\mathsf{X}$ of size $O(\frac{n}{\tau})$
where $\mathsf{X}[\floor{\SSA_P[i]/\tau}] = i$, 
i.e., $\ISA_P[i] = X[\floor{i/\tau}]$ for all $i\in P$.
By directly applying the algorithm of Kasai et al.~\cite{kasai01:_linear_time_longes_common_prefix},
$\SLCP_P$ can be computed from $\SSA_P$ and (the representation of) $\ISA_P$
in $O(n)$ time and $O(\frac{n}{\tau})$ space.

\section{Data Structure and Query Computation}
\label{sec:new}
Our algorithms are based on the same observation as
used in~\cite{bille15:_longes_common_exten_sublin_space}.

\begin{observation}[\cite{bille15:_longes_common_exten_sublin_space}]
  \label{observation:main}
  For any positions $i,j,k \in [0..n-1]$
  if $\lcp(j,k) \geq \lcp(i,j)$ then,
  $\lcp(i,j) = \min \{ \lcp(i,k), \lcp(j,k) \}$.
\end{observation}

The observation allows us to reduce the computation of LCP values between a pair of positions,
to the computation of LCP values between a different pair that are in some subset
of positions.
For each specific position $i$ called sampled positions,
and for each such subset $S$,
a position $\pi(i,S) = \arg \max_{i^{\prime} \in S} \{\lcp(i,i^{\prime})\}$ is precomputed.
The idea is that the size of $S$ gets smaller after each reduction,
therefore giving a bound on the query time.

\begin{corollary}\label{corollary:main}
  For any pair of positions $i \in S\subseteq [0..n-1]$ and $j \in [0..n-1]$,
  $\lcp(i,j) = \min \{ \lcp(i,\pi(j,S)),$ $\lcp(j,\pi(j,S)) \}$.
\end{corollary}

\subsection{Bille et al.'s Data Structure~\cite{bille15:_longes_common_exten_sublin_space}}
\label{subsection:billetal}
We first introduce a slightly modified version of the deterministic data structure
by Bille et al.~\cite{bille15:_longes_common_exten_sublin_space}
that uses $O(\frac{n}{\tau})$ space and allows queries in
$O(\tau\log^2(\frac{n}{\tau}))$ time,
where $\tau$ is a parameter in the range $1 \leq \tau \leq n$.
We note that the modifications do not affect the asymptotic complexities.

Let $t = \tau\ceil{\log \frac{n}{\tau}}$, and $\mathcal{P} = [0..n-1]^t_p$, where $p = (n-1) \bmod t$,
be the set of positions called {\em sampled positions}.
The data structure of~\cite{bille15:_longes_common_exten_sublin_space} to
answer $\lcp(i,j)$ for any $0 \leq i < j < n$
consists of two main parts, one for when $j-i \geq t$, and 
the other for when $j - i < t$.
Since we will use the latter part as is, we will only describe the former.
The query time, space, and preprocessing time of the latter part are respectively,
$O(\tau\log\frac{n}{\tau})$, $O(\frac{n}{\tau})$, and $O(n)$
(See Section 2 of~\cite{bille15:_longes_common_exten_sublin_space}).

Consider a full binary tree where the root corresponds to the interval $[0..n-1]$,
and for any node, the left and right children split their parent interval almost evenly,
but assuring that the right-most position in the left child is a sampled position.
Thus, there will be $\ceil{n/t}$ leaves corresponding to intervals of size $t$
(except perhaps for the left most leaf which may be smaller),
and the height of the tree is $O(\log\frac{n}{t})$.
For any internal node $v$ in the tree with interval $I_v$, let $I_{\ell(v)}$ and $I_{r(v)}$
respectively be its left and right children.
For all sampled positions $i \in I_{r(v)}\cap \mathcal{P}$,
a position
$\pi(i,I_{\ell(v)}) = \arg\max_{i^\prime \in I_{\ell(v)}} \{ \lcp(i, i^{\prime}) \}$
and $L(i,I_{\ell(v)}) = \lcp(i,\pi(i,I_{\ell(v)}))$
are computed and stored.
The size of the data structure is therefore
$O(\frac{n}{t}\log\frac{n}{t}) = O(\frac{n}{\tau})$.

Assume w.l.o.g. that $j > i$.
A query for $\lcp(i,j)$ with $j-i \geq t$ is computed as follows.
First, 
compare up to $\delta < t$ characters of $w[i..]$ and $w[j..]$
until we encounter a mismatch, in which case we obtain an answer,
or $j+\delta$ is a sampled position.
Let
$I_v$ be the interval such that $i+\delta \in I_{\ell(v)}$ and $j+\delta \in I_{r(v)}$.
From the preprocessing, we obtain a position $\pi(j+\delta,I_{\ell(v)}) \in I_{\ell(v)}$, 
which, from Corollary~\ref{corollary:main}, gives:
\begin{eqnarray*}
  \lcp(i,j) &=& \delta + \lcp(i+\delta,j+\delta)\\
  &=& \delta + \min \{ \lcp(i+\delta, \pi(j+\delta,I_{\ell(v)})),
                       \lcp(j+\delta, \pi(j+\delta,I_{\ell(v)})) \} \\
  &=& \delta + \min \{ \lcp(i+\delta, \pi(j+\delta,I_{\ell(v)})),
                        L(j+\delta, I_{\ell(v)}) \} 
\end{eqnarray*}
Thus, the problem can be reduced to computing $\lcp(i+\delta,\pi(j,I_{\ell(v)}))$,
where both $i+\delta,\pi(j,I_{\ell(v)}) \in I_{\ell(v)}$, and we apply the algorithm
recursively.
Note that if $j\in I_{r(v)}$ we have, from the definition of the intervals,
that $j+\delta\in I_{r(v)}$, so each recursion takes us further down the tree.
When an interval corresponding to a leaf node is reached,
we have that $j-i < t$ and use the other data structure
(for a description of which we refer the reader to~\cite{bille15:_longes_common_exten_sublin_space}).
Since we compare up to $t$ characters at each height, the total query time is
$O(t \log \frac{n}{t}) = O(\tau\log^2\frac{n}{\tau})$.

\subsection{New Data Structure}
\label{subsection:newdatastructure}
Let $t = \tau\ceil{\log\tau}$, $p = (n-1) \bmod t$,
and let $\mathcal{P} = [0..n-1]^t_p$ be the set of sampled positions.
Instead of considering a hierarchy of intervals of positions,
we classify the positions according to
their distance to the closest sampled position to their right.
Define 
$S_k = \{ i \mid (i + d) \bmod t = p, d \in ([2^{k-1}..2^k -1]\cap[1..t-1]) \}$
for $k = 1, \ldots, \ceil{\log t}$ (see also Figure~\ref{fig:S_k}).
\begin{figure}
\begin{center}
\includegraphics[width=0.6\linewidth]{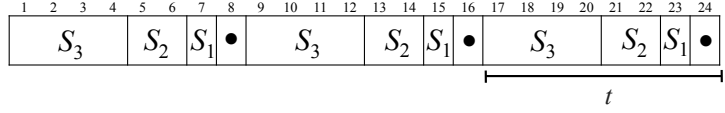}
\caption{Examples of $S_k$ with $k = 1, 2, 3$ for the sampled positions
specified by black dots.}
\label{fig:S_k}
\end{center}
\end{figure}
The preprocessing computes and stores for each sampled position $i \in \mathcal{P}$
and each $S_k$, a position
$\pi(i,S_k) = \arg\max_{i^{\prime} \in S_k} \{ \lcp(i, i^{\prime}) \}$,
and $L(i,S_k) = \lcp(i,\pi(i,S_k))$.
Also, $\SLCP_\mathcal{P}$ is computed and preprocessed for range minimum queries so that
for any $i,j \in \mathcal{P}$, $\lcp(i,j)$ can be computed in constant time.
Thus, the space required for the data structure
is $O(\frac{n}{t}\log t) = O(\frac{n}{\tau})$.

A query $\lcp(i,j)$ is computed as follows.
First, compare up to $\delta$ characters of
$w[i..]$ and $w[j..]$ until we encounter a mismatch, in which case we obtain an answer,
or, either $i+\delta$ or $j+\delta$ is a sampled position.
If both $i+\delta$ and $j+\delta$ are sampled positions,
$\lcp(i,j) = \delta + \lcp(i+\delta,j+\delta)$ can be answered
in constant time.
Assume w.l.o.g. that only $j+\delta$ is a sampled position, and 
let $k$ be such that $i+\delta\in S_{k}$.
Then, from Corollary~\ref{corollary:main} and the preprocessing,
we have
\begin{eqnarray*}
  \lcp(i,j) &=& \delta + \lcp(i+\delta,j+\delta)\\
  &=& \delta + \min \{ \lcp(i+\delta, \pi(j+\delta,S_k)),
                       \lcp(j+\delta, \pi({j+\delta,S_k})) \} \\
  &=& \delta + \min \{ \lcp(i+\delta, \pi(j+\delta,S_k)),
                       L(j+\delta,S_k) \} 
\end{eqnarray*}
and the problem has been reduced to computing
$\lcp(i+\delta, \pi(j+\delta,S_k))$ where both $i+\delta,\pi(j+\delta,S_k) \in S_{k}$,
and the processes are repeated.
Notice that in the next step, at least $2^{k-1}$ characters are compared
until one of the two positions reaches a sampled position.
This implies that the remaining distance to the closest sampled position
of the other position will be at most $2^{k-1}-1$, and thus will be in $S_{k^{\prime}}$
for some $k^{\prime} \leq k-1$.
Therefore, the process will only be repeated at most $\ceil{\log t}$ times.
Because the number of characters compared in each step is bounded by
$t$ and is at least halved every step,
the total number of character comparisons and thus the query time is $O(t) = O(\tau\log\tau)$.

\subsection{Combining the structures}
We can combine the structures described in Sections~\ref{subsection:billetal} and~\ref{subsection:newdatastructure},
to achieve $O(\tau\log\frac{n}{\tau})$ query time using $O(\frac{n}{\tau})$ space.

Let $t = \tau\ceil{\log \frac{n}{\tau}}$, and $\mathcal{P} = [0..n-1]^t_p$, where $p = (n-1) \bmod t$,
be the set of positions called {\em sampled positions}.
We consider both the structures described in Section~\ref{subsection:billetal} and Section~\ref{subsection:newdatastructure},
with the following modifications.
Let $d_t = 2^{\ceil{\log t} - {\ceil{\log\frac{n}{t}}}} = O(\frac{t^2}{n})$.
For Bille et al.'s structure, we make two modifications.
First, we modify the data structure so that
for each node $I_v$ and sampled position $i \in I_{r(v)}\cap \mathcal{P}$,
we only consider points that are at most $d_t$ from the closest sampled position to the right,
i.e.,
instead of $\pi(i,I_{\ell(v)})$ and $L(i,I_{\ell(v)})$,
we compute and store $\pi(i,I_{\ell(v)}\cap D) = \arg\max_{i^\prime \in I_{\ell(v)}\cap D} \{ \lcp(i, i^{\prime}) \}$
and $L(i,I_{\ell(v)}\cap D)$,
where $D = \{ i^\prime \mid (i^\prime+d)\bmod t = p, 0 \leq d < d_t \}$.
In addition to this, we compute and store for all sampled position
$i \in I_{\ell(v)}\cap \mathcal{P}$,
a position
$\pi(i,I_{r(v)}\cap D) = \arg\max_{i^\prime \in I_{r(v)}\cap D} \{ \lcp(i, i^{\prime}) \}$
and $L(i,I_{r(v)}\cap D) = \lcp(i,\pi(i,I_{r(v)}\cap D))$.
This will only double the total size of the structure and thus the space remains $O(\frac{n}{\tau})$.
For the new data structure,
we keep the definition of $\pi(i,S_k)$ and $L(i,S_k)$,
but only for values $k = \ceil{\log t} - \ceil{\log\frac{n}{t}}, \ldots, \ceil{\log t}$.
Thus, although the value of $t$ has changed,
the total size of the data structure is still
$O(\frac{n}{t}\log\frac{n}{t}) = O(\frac{n}{\tau})$. 

Queries $\lcp(i,j)$ are answered as follows:
first use the new data structure recursively using the same algorithm until the problem is
reduced to a query between a sampled position and another position not in any
$S_k~(k\in [\ceil{\log t} - \ceil{\log\frac{n}{t}}..\ceil{\log t}])$.
This means that the distance from each of the query positions to the closest sampled position is at most $d_t$.
The total number of character comparisons conducted is $O(t) = O(\tau\log\frac{n}{\tau})$.
Then, we switch to Bille et al.'s structure using the same algorithm with the exception
of when comparing up to $\delta$ characters of $w[i..]$ and $w[j..]$, we continue until
{\em either} $i+\delta$ {\em or} $j+\delta$ (instead of just $j+\delta$)
is a sampled position.
Since the distance to the closest sampling position is at most $O(\frac{t^2}{n})$ and
by definition of $\pi(i,I_{\ell(v)}\cap D)$ and $\pi(i,I_{r(v)}\cap D)$,
we have that this condition holds for all following recursions.
Thus, at most $O(\frac{t^2}{n})$ character comparisons will be conducted for each height,
for a total of $O(\frac{t^2}{n}\log \frac{n}{t}) = O(t) = O(\tau\log\frac{n}{\tau})$.\footnote{Letting $x = \frac{n}{t}$, $O(\frac{t^2}{n}\log\frac{n}{t}) = O(t\frac{\log x}{x}) = O(t)$.}

\section{Building the Structures}
\label{sec:construct}
Bille et al.~\cite{bille15:_longes_common_exten_sublin_space}
describe a preprocessing that runs in
$O(n^2)$ time\footnote{However, we believe the analysis in Section 2.5 of~\cite{bille15:_longes_common_exten_sublin_space} is not entirely correct; although the size of $|I|$ is halved at each level their numbers double, and so the time complexity should be $O(n\cdot n + n\cdot(n/2)\cdot 2 \cdots + n\cdot(n/t)\cdot t) = O(n^2\log\frac{n}{t})$ time. Also, they assume that the evenly spaced sparse suffix array can be constructed in $O(n)$ time and $O(n/\tau)$ space, for the integer alphabet. However, the paper they cite assumes a constant size alphabet and to the best of our knowledge, we do not know of an algorithm achieving such space/time trade-off.}
and $O(\frac{n}{\tau})$ space.
Here, we show that this can be reduced to
$O(\tau n + n\log n)$ time using the same space.
While the algorithm of~\cite{bille15:_longes_common_exten_sublin_space}
builds the sparse suffix array containing only the suffixes starting at
sampled positions and applies pattern matching,
our trick is to build a sparse suffix array and sparse LCP array that includes
other suffixes as well, in several (namely $\tau$) rounds,
so that the suffix with maximum LCP with respect to each sampled position
can be found by scans on the suffix array.

For integer alphabets, sparse suffix arrays and sparse LCP arrays can be constructed in
$O(n)$ time if $O(n)$ space is allowed, simply by first building the (normal)
suffix array and LCP array and removing the unwanted elements.
For constant size alphabets, the evenly spaced sparse suffix array and sparse LCP array
with sampling rate $\tau$
can be constructed in $O(n)$ time and $O(\frac{n}{\tau})$
space~\cite{karkkainen96:_spars_suffix_trees}.
However, when the alphabet size $\sigma$ is not constant,
this is $O(n\log\sigma)$ time and $O(\frac{n}{\tau})$ space,
since the computation is based on character comparisons.
(Notice that linear time algorithms for computing the suffix array
for the meta string will not achieve $O(n)$ time and $O(\frac{n}{\tau})$ space,
since the use of radix sort requires at least $O(\sigma)$ space for the buckets.)
If this is repeated $\tau$ times, this results in $O(n\tau\log\sigma)$ time using
$O(\frac{n}{\tau})$ space.

We first describe a technique to 
compute the sparse suffix array and LCP array that contains two sets of evenly spaced suffixes,
namely for offsets $p$ and $q$,
and
to repeat this $\tau$ times,
namely for offsets $p = (n-1)\bmod \tau$ and $q = (n-1)\bmod\tau,\ldots,(n-\tau)\bmod\tau$,
so that the total time for their construction is $O(n\tau)$ time using $O(\frac{n}{\tau})$ space.
Then, we describe the construction of the data structures of Section~\ref{sec:new}
using this technique.

\subsection{Common Tools}
\label{subsection:commontools}
For any string (or meta-string) $w$ and $0 \leq i < |w|$,
let $\CA{w}$ denote an array containing a permutation of
$[0..|w|-1]$ such that
$w[\CA{w}[i]] \leq w[\CA{w}[j]]$ for any $0 \leq i < j < |w|$, i.e.,
$\CA{w}$ is an array of positions sorted according to the
character at each position.
(Note that $\CA{w}$ is not necessarily unique.)

\begin{lemma}\label{lemma:sortingmetacharacters}
  For any string $w$ and $0 \leq p < \tau$,
  $\CA{\Meta{w}{\tau}{p}}$
  can be computed in
  $O(n\tau)$ time using
  $O(\frac{n}{\tau})$ space.
\end{lemma}
\begin{proof}
  Since each character of $w$ can be represented in $O(\log n)$ bits,
  the length of each meta-character of $\Meta{w}{\tau}{p}$ is $O(\tau\log n)$ bits.
  We simply use LSD radix sort with a bucket size of $\frac{n}{\tau}$,
  i.e., we bucket sort using $\log(n/\tau)$ bits at a time.
  Thus, $O(\frac{\tau\log n}{\log(n/\tau)})$ rounds of bucket sort
  is conducted on $\frac{n}{\tau}$ items, resulting in
  $O(\frac{n \log n}{\log(n/\tau)}) = O(n\tau)$\footnote{
    We can assume that $\log (n/\tau) \geq 1$.
    If $\tau \leq \sqrt{n}$, then $\frac{\log n}{\log (n/\tau)} \leq 2 = O(\tau)$,
    otherwise, $\frac{\log n}{\log (n/\tau)} < \log n = O(\tau)$.
  }, giving the result.
\end{proof}

\begin{lemma}\label{lemma:sortingnextmetacharacters}
  For any string $w$ and $0 \leq p < \tau$,
  $\CA{\Meta{w}{\tau}{p}}$
  can be computed from
  $\CA{\Meta{w}{\tau}{p^\prime}}$, where $p^\prime = (p+1)\bmod\tau$,
  in
  $O(n)$ time and $O(\frac{n}{\tau})$ space.
\end{lemma}
\begin{proof}
  We simply continue the LSD radix sort,
  and do an extra $O(\frac{\log n}{\log(n/\tau)}) = O(\tau)$ rounds of bucket sort
  for the preceding character of each meta-character.
\end{proof}

\begin{lemma}\label{lemma:lcp_for_twomods}
  For any string $w$, $0 \leq p,q < \tau$,
  let $P = [0..n-1]^\tau_p$ and $Q = [0..n-1]^\tau_q$.
  Given $\CA{\Meta{w}{\tau}{p}}$
  and $\CA{\Meta{w}{\tau}{q}}$,
  $\SSA_{P\cup Q}$ and $\SLCP_{P\cup Q}$
  can be computed in $O(n)$ time using $O(\frac{n}{\tau})$ space.
\end{lemma}
\begin{proof}
  We first compute
  $\CA{w^\prime}$ for meta-string
  $w^\prime = \Meta{w}{\tau}{p}0\Meta{w}{\tau}{q}$.
  This can be done in $O(n)$ time and $O(\frac{n}{\tau})$ space
  by merging
  $\CA{\Meta{w}{\tau}{p}}$
  and $\CA{\Meta{w}{\tau}{q}}$,
  (and adding $|\Meta{w}{\tau}{p}0|$ to entries in $\CA{\Meta{w}{\tau}{q}}$)
  since each comparison of meta characters can be done in $O(\tau)$ time.
  Using $\CA{w^\prime}$, we then rename the characters of $w^\prime$ and create a string
  $w^*$ such that $w^*[i] = |\{ w^\prime[j] \mid w^\prime[j] < w^\prime[i], 0 \leq j < |w^\prime|\}| + 1$,
  in $O(n)$ time and $O(\frac{n}{\tau})$ space.
  Since $w^*$ consists of integers bounded by its length, we can apply any linear time suffix
  sorting algorithm and compute $\SA_{w^*}$ and $\LCP_{w^*}$ in
  $O(\frac{n}{\tau})$ time and space.
  As the lexicographic order of suffixes of $w^*$
  (except for $\SSA_{w^*}[0] = |\Meta{w}{\tau}{p}|$)
  correspond to the
  lexicographic order of suffixes of $w$ that start at positions in $P\cup Q$,
  we can obtain $\SSA_{P\cup Q}$ from $\SA_{w^*}$ by appropriately translating the indices.
  More precisely, 
  for $1 \leq i < |w^\prime|$, let $\SSA_{w^*}[i] = j$.
  If $0 \leq j < |\Meta{w}{\tau}{p}|$,
  then $\SSA_{P\cup Q}[i-1] = j \tau + p$,
  and otherwise (if
  $|\Meta{w}{\tau}{p}0| \leq j < |w^\prime|$), then
  $\SSA_{P\cup Q}[i-1] = (j - |\Meta{w}{\tau}{p}0|)\tau + q$.
  We can also obtain $\SLCP_{P\cup Q}$ from $\LCP_{w^*}$
  by multiplying a factor of $\tau$
  and doing up to $\tau$ character comparisons per pair of adjacent suffixes
  in the suffix array, in a total of $O(n)$ time.
\end{proof}
\begin{corollary}
  \label{corollary:successive_ssa}
  For any string $w$, let $p = n\bmod\tau$.
  $\SSA_{P\cup Q}$ and $\SLCP_{P\cup Q}$ can be computed successively
  for each $q = p, (p-1)\bmod\tau, \ldots, (p-\tau+1)\bmod\tau$, where
  $P = [0..n-1]^\tau_p$ and $Q = [0..n-1]^\tau_q$,
  in $O(n\tau)$ time using $O(\frac{n}{\tau})$ space.
\end{corollary}
\begin{proof}
  For $p = q$, we first compute
  $\CA{\Meta{w}{\tau}{p}} = \CA{\Meta{w}{\tau}{q}}$ using Lemma~\ref{lemma:sortingmetacharacters}.
  By applying Lemma~\ref{lemma:sortingnextmetacharacters},
  we can successively compute $\CA{\Meta{w}{\tau}{q}}$
  for $q = (p-1)\bmod\tau,\ldots,(p-\tau+1)\bmod\tau$.
  Thus, with Lemma~\ref{lemma:lcp_for_twomods}, we can successively
  compute $\SSA_{P\cup Q}$ and $\SLCP_{P\cup Q}$
  in $O(n\tau)$ total time and $O(\frac{n}{\tau})$ space.
\end{proof}

\subsection{Faster Construction of Bille et al.'s Data Structure}
\label{subsection:construction_bille}
We show that Bille et al.'s data structure can be constructed in
$O(n\tau + n\log\frac{n}{\tau})$ using $O(\frac{n}{\tau})$ space.
Let $p = (n-1)\bmod\tau$. Using Corollary~\ref{corollary:successive_ssa}, we successively compute
$\SSA_{P\cup Q}$ and $\SLCP_{P\cup Q}$ for each $q = p, (p-1)\bmod\tau, \ldots, (p-\tau+1)\bmod\tau$, where
$P = [0..n-1]^\tau_p$ and $Q = [0..n-1]^\tau_q$. This can be done in a total of $O(n\tau)$ time,
and $O(\frac{n}{\tau})$ space.
Recall that
$t = \tau\ceil{\log \frac{n}{\tau}}$, and $\mathcal{P} = [0..n-1]^t_p$, where $p = (n-1) \bmod t$.
Since $t$ is a multiple of $\tau$, we have $\mathcal{P} \subseteq P$.

For each $q$ we do the following.
$\SLCP_{P\cup Q}$ is preprocessed in $O(\frac{n}{\tau})$ time and space to
answer RMQ in constant time, thus allowing us to compute
$\lcp(i,j)$ for any $i,j \in P\cup Q$ in constant time.
For any interval $I_v\subseteq [0..n-1]$ corresponding to a node in the binary tree
let $I^q_v = I_v\cap ( P \cup Q )$.
Note that for $I_{\mathsf{root}} = [0..n-1]$,  $\SSA_{I_{\mathsf{root}}^q} = \SSA_{P\cup Q}$.
Now, for any node $I_v$, assume that $\SSA_{I_v^q}$ is already computed.
By simple linear time scans on $\SSA_{I_v^q}$, we can obtain, for each
sampled position $i  = \SSA_{I_v^q}[x] \in I_{r(v)}^q\cap \mathcal{P}$,
the two suffixes $\SSA_{I_v^q}[j^-],\SSA_{I_v^q}[j^+] \in I_{\ell(v)}^q\cap Q$
which are lexicographically closest to $i$,
i.e.,
$j^- = \max \{ j < x \mid \SSA_{I_v^q}[j] \in I_{\ell(v)}^q\cap Q\}$,
$j^+ = \min \{ j > x \mid \SSA_{I_v^q}[j] \in I_{\ell(v)}^q\cap Q\}$, if they exist.
Then, the longer of $\lcp(i,\SSA_{I_v^q}[j^-])$ and $\lcp(i,\SSA_{I_v^q}[j^+])$ 
gives
$\pi(i,I_{\ell(v)}^q\cap Q) = \arg\max_{i^{\prime} \in I_{\ell(v)}^q\cap Q} \{ \lcp(i,i^{\prime}) \}$
and $L(i,I_{\ell(v)}^q\cap Q) = \lcp(i, \pi(i,I_{\ell(v)}^q\cap Q))$.
Since $i,\SSA_{I_v^q}[j^+],\SSA_{I_v^q}[j^-] \in P\cup Q$,
these values can be computed in constant time,
and thus can be computed in $O(|I_v^q|)$ total time for all sampled positions $i\in I_{r(v)}^q\cap \mathcal{P}$.
Next, for the child intervals, $\SSA_{I_{\ell(v)}^q}$ and $\SSA_{I_{r(v)}^q}$
can be computed in $O(|I_v^q|)$ time by a simple scan on $\SSA_{I_v^q}$, and the
computation is recursed for each child.
Since the union of $I_v^q\cap Q$ over all $q$ is $I_{v}$,
we have 
$\pi(i,I_{\ell(v)}) = \pi(i,I_{\ell(v)}^{\hat{q}})$
and $L(i,I_{\ell(v)}) = L(i,I_{\ell(v)}^{\hat{q}})$,
where $\hat{q} = \arg\max_{0\leq q^\prime < \tau} \{ \lcp(i, \pi(i,I_{\ell(v)}^{q^\prime}\cap Q)) \}$,
so we can obtain $\pi(i,I_{\ell(v)})$ and $L(i,I_{\ell(v)})$ 
for each sampled position $i$ and interval $I_v$
by repeating the above process for each $q$.

Since the processing at each node is linear in the size of the arrays
whose total size at a given depth is $O(\frac{n}{\tau})$,
the total time for the recursion is $O(\frac{n}{\tau}\log\frac{n}{\tau})$ for each $q$.
Thus in total, the preprocessing can be done in $O(n\tau + n\log\frac{n}{\tau})$ time.

\begin{theorem}
  For any string of length $n$ and integer $1 \leq \tau \leq n$,
  a data structure of size $O(n/\tau)$ can be constructed in
  $O(n\tau + n\log\frac{n}{\tau})$ time using $O(\frac{n}{\tau})$ space,
  such that
  for any $0 \leq i,j < n$, $\lcp(i,j)$ can be answered in $O(\tau\log^2\frac{n}{r})$ time.
\end{theorem}

\subsection{Fast Construction of New Data Structure}
\label{subsection:construction_new}
Let $p = (n-1)\bmod\tau$. Using Corollary~\ref{corollary:successive_ssa}, we successively compute
$\SSA_{P\cup Q}$ and $\SLCP_{P\cup Q}$ for each $q = p, (p-1)\bmod\tau, \ldots, (p-\tau+1)\bmod\tau$, where
$P = [0..n-1]^\tau_p$ and $Q = [0..n-1]^\tau_q$. This can be done in a total of $O(n\tau)$ time,
and $O(\frac{n}{\tau})$ space.
Recall that
$t = \tau\ceil{\log\tau}$, and $\mathcal{P} = [0..n-1]^t_p$, where $p = (n-1) \bmod t$.
Since $t$ is a multiple of $\tau$, we have $\mathcal{P} \subseteq P$.

For each $q$ we do the following.
$\SLCP_{P\cup Q}$ is preprocessed in $O(\frac{n}{\tau})$ time and space to
answer RMQ in constant time, thus allowing us to compute
$\lcp(i,j)$ for any $i,j \in P\cup Q$ in constant time.
Let $S_k^q = S_k\cap Q$ for any $1\leq k\leq\ceil{\log t}$.
Next, we conduct for each $k = 1,\ldots,\ceil{\log t}$,
linear time scans on $\SSA_{P\cup Q}$ so that
for each sampled position $i = \SSA_{P\cup Q}[x] \in \mathcal{P}$,
the two suffixes $\SSA_{P\cup Q}[j^-],\SSA_{P\cup Q}[j^+] \in S_k^q$
which are lexicographically closest to $i$, i.e.,
$j^- = \max \{ j < x \mid \SSA_{P\cup Q}[j] \in S_k^q\}$,
$j^+ = \min \{ j > x \mid \SSA_{P\cup Q}[j] \in S_k^q\}$, if they exist.
Then, the longer of $\lcp(i,\SSA_{P\cup Q}[j^-])$ and $\lcp(i,\SSA_{P\cup Q}[j^+])$ 
gives
$\pi(i,S_k^q) = \arg\max_{i^{\prime} \in S_k^q} \{ \lcp(i,i^{\prime}) \}$.
Since $i,\SSA_{P\cup Q}[j^+],\SSA_{P\cup Q}[j^-] \in P\cup Q$,
these values can be computed in constant time,
resulting in a total of $O(\frac{n}{\tau}\log\tau)$ time for all $i$ and $k$.
Since the union of $S_k^q$ over all $q$ is $S_k$,
we have 
$\pi(i,S_k) = \pi(i,S_k^{\hat{q}})$ and $L(i,S_k) = L(i,S_k^{\hat{q}})$,
where $\hat{q} = \arg\max_{0\leq q^\prime < \tau} \{ \lcp(i, \pi(i,S_k^{q^\prime})) \}$,
so we can obtain $\pi(i,S_k)$ and $L(i,S_k)$ 
for each sampled position $i$ and $S_k$
by repeating the above process for each $q$, taking $O(n\log\tau)$ time.
Thus, the total time for preprocessing is dominated by Corollary~\ref{corollary:successive_ssa},
and is $O(n\tau)$.

\begin{theorem}[Fast Construction of New Data Structure]\label{theorem:fastnew}
  For any string of length $n$ and integer \hbnote*{is this optimal?}{$1 \leq \tau \leq \frac{n}{\log n}$},
  a data structure of size $O(n/\tau)$ can be constructed in
  $O(n\tau)$ time using $O(\frac{n}{\tau})$ space,
  such that
  for any $0 \leq i,j < n$, $\lcp(i,j)$ can be answered in $O(\tau\log\tau)$ time.
\end{theorem}

\subsection{Fast Construction of Combined Data Structure}
The construction of the combined data structure is done using the same algorithms as described
in Sections~\ref{subsection:construction_bille} and~\ref{subsection:construction_new}, but with minor modifications.
For Bille et al.'s data structure, we only need to consider in addition to sampled positions,
the positions in $D = \{ i^\prime \mid (i^\prime+d)\bmod t = p, 0 \leq d < d_t \}$
due to the modification introduced for the combination.
This reduces the array sizes (and thus the computation time) needed for the computation
of $\pi(i,I_{\ell(v)})$ and $\pi(i,I_{r(v)})$ (and $L(i,I_{\ell(v)})$ and $L(i,I_{r(v)})$)
to
$O(\frac{n}{t} + \frac{n}{t}\cdot\frac{t^2}{n}\cdot\frac{1}{\tau}) = O(\frac{n}{t} + \frac{t}{\tau}) = O(\frac{n}{\tau\log\frac{n}{\tau}} + \log\frac{n}{\tau})$
for a total of
$O(\frac{n}{\tau} + \log^2\frac{n}{\tau})$
for all depths, and for all $q$, we get
$O(n + \tau\log^2\frac{n}{\tau}) = O(n + n\frac{\log^2\frac{n}{\tau}}{\frac{n}{\tau}})) = O(n)$.
Thus, the total time for preprocessing is now dominated by Corollary~\ref{corollary:successive_ssa},
and is $O(n\tau)$.

\begin{theorem}\label{theorem:combined}
  For any string of length $n$ and integer $1 \leq \tau \leq n$,
  a data structure of size $O(n/\tau)$ can be constructed in
  $O(n\tau)$ time using $O(\frac{n}{\tau})$ space,
  such that
  for any $0 \leq i,j < n$, $\lcp(i,j)$ can be answered in $O(\tau\log\frac{n}{\tau})$ time.
\end{theorem}

Since
$\tau \leq \frac{n}{\tau}$ when $\tau \leq \sqrt{n}$, and
$\tau \geq \frac{n}{\tau}$ when $\tau \geq \sqrt{n}$, we get the following
by simply choosing the data structure of
Theorems~\ref{theorem:fastnew} and~\ref{theorem:combined},
depending on the value of $\tau$.
\begin{corollary}
  For any string of length $n$ and integer $1 \leq \tau \leq n$,
  a data structure of size $O(n/\tau)$ can be constructed in
  $O(n\tau)$ time using $O(\frac{n}{\tau})$ space,
  such that
  for any $0 \leq i,j < n$, $\lcp(i,j)$ can be answered in
  $O(\tau\min\{\log\tau, \log\frac{n}{\tau}\})$ time.
\end{corollary}

\section{Applications}
\label{sec:appl}
Using the proposed data structure,
the lexicographic order between two arbitrary suffixes can be computed in
$O(\tau\min\{\log\tau,\log\frac{n}{\tau}\})$ time using $O(\frac{n}{\tau})$ space.
Thus, using any $O(n\log n)$ comparison based sorting algorithm,
we can compute the suffix array of a string of length $n$ in
$O(\min\{\log\tau,\log\frac{n}{\tau}\}n\tau\log n)$ time
using $O(\frac{n}{\tau})$ working space, excluding the input and output.
The best known deterministic space/time trade-off is
$O(n\tau^2)$ time (for $1 \leq \tau \leq \sqrt[4]{n}$)
using the same space~\cite{Karkkainen:2006je},
and our algorithm is better
when $\tau = \Omega(\log^{1+\epsilon} n)$ for any $\epsilon > 0$.

\section*{Acknowledgements}
HB,SI,MT were supported by JSPS KAKENHI Grant Numbers
25280086, 26280003, 25240003.

\bibliographystyle{plain}
%\bibliography{ref}

\end{document}